\documentclass[twoside]{article}%
\usepackage{amssymb}
\usepackage{amsfonts}
\usepackage{amsmath}
\usepackage{graphicx}%
\providecommand{\U}[1]{\protect\rule{.1in}{.1in}}
\newtheorem{theorem}{Theorem}

\newtheorem{corollary}[theorem]{Corollary}

\newtheorem{lemma}[theorem]{Lemma}

\newtheorem{remark}[theorem]{Remark}

\newenvironment{proof}[1][Proof]{\noindent\textbf{#1.} }{\ \rule{0.5em}{0.5em}}
\ifx\pdfoutput\relax\let\pdfoutput=\undefined\fi
\newcount\msipdfoutput
\ifx\pdfoutput\undefined\else
\ifcase\pdfoutput\else
\ifx\paperwidth\undefined\else
\ifdim\paperheight=0pt\relax\else\pdfpageheight\paperheight\fi
\ifdim\paperwidth=0pt\relax\else\pdfpagewidth\paperwidth\fi
\fi\fi\fi
\begin{document}

\title{\textbf{Analytical} \textbf{Blowup Solutions to the Pressureless
Navier-Stokes-Poisson Equations with Density-dependent Viscosity in }$R^{N}$}
\author{Y\textsc{uen} M\textsc{anwai\thanks{E-mail address: nevetsyuen@hotmail.com }}\\\textit{Department of Applied Mathematics, }\\\textit{The Hong Kong Polytechnic University,}\\\textit{Hung Hom, Kowloon, Hong Kong}}
\date{Revised 24-Apr-2009}
\maketitle

\begin{abstract}
We study the N-dimensional pressureless Navier--Stokes-Poisson equations with
density-dependent viscosity. With the extension of the blowup solutions for
the Euler-Poisson equations, the analytical blowup solutions, in radial
symmetry, in $R^{N}$ ($N\geq2$) are constructed.

\end{abstract}

\section{Introduction}

The evolution of a self-gravitating fluid can be formulated by the
Navier-Stokes-Poisson equations of the following form:
\begin{equation}
\left\{
\begin{array}
[c]{rl}%
{\normalsize \rho}_{t}{\normalsize +\nabla\bullet(\rho\vec{u})} &
{\normalsize =}{\normalsize 0,}\\
{\normalsize (\rho\vec{u})}_{t}{\normalsize +\nabla\bullet(\rho\vec{u}%
\otimes\vec{u})+\nabla P} & {\normalsize =}{\normalsize -\rho\nabla
\Phi+vis(\rho,\vec{u}),}\\
{\normalsize \Delta\Phi(t,x)} & {\normalsize =\alpha(N)}{\normalsize \rho,}%
\end{array}
\right.  \label{Euler-Poisson}%
\end{equation}
where $\alpha(N)$ is a constant related to the unit ball in $R^{N}$:
$\alpha(1)=2$; $\alpha(2)=2\pi$ and For $N\geq3,$%
\begin{equation}
\alpha(N)=N(N-2)V(N)=N(N-2)\frac{\pi^{N/2}}{\Gamma(N/2+1)},
\end{equation}
where $V(N)$ is the volume of the unit ball in $R^{N}$ and $\Gamma$ is a Gamma
function. And as usual, $\rho=\rho(t,x)$ and $\vec{u}=\vec{u}(t,x)\in
\mathbf{R}^{N}$ are the density, the velocity respectively. $P=P(\rho)$\ is
the pressure. The evolution of the cosmology can be modelled by the dust
distribution without pressure term. That describes the stellar systems of
collisionless, gravitational N-body systems.

In the above system, the self-gravitational potential field $\Phi=\Phi
(t,x)$\ is determined by the density $\rho$ through the Poisson equation.

And $vis(\rho,\vec{u})$ is the viscosity function:%
\begin{equation}
vis(\rho,\vec{u}):=\bigtriangledown(\mu(\rho)\bigtriangledown\bullet\vec{u}).
\label{pp1}%
\end{equation}
In this article, we seek the radial solutions
\begin{equation}
\rho(t,x)=\rho(t,r)\text{ and }\vec{u}=\frac{\vec{x}}{r}V(t,r):=\frac{\vec{x}%
}{r}V\text{,}%
\end{equation}
with $r=\left(  \sum_{i=1}^{N}x_{i}^{2}\right)  ^{1/2}.$ In particular, we can
deduce the viscosity function (\ref{pp1}) to be%
\begin{align}
\bigtriangledown(\mu(\rho)\bigtriangledown\bullet\vec{u})  &  =\nabla\left[
\mu(\rho)\left(  \left(  \frac{\partial}{\partial x_{1}},\frac{\partial
}{\partial x_{2}},...,\frac{\partial}{\partial x_{N}}\right)  \cdot\left(
\frac{x_{1}}{r},\frac{x_{2}}{r},...,\frac{x_{N}}{r}\right)  V\right)  \right]
\\
&  =\nabla\left[  \mu(\rho)\overset{N}{\underset{k=1}{\sum}}\frac{\partial
}{\partial x_{k}}\left(  \frac{x_{k}}{r}V\right)  \right] \\
&  =\nabla\left[  \mu(\rho)\overset{n}{\underset{k=1}{\sum}}\left(
\frac{r^{2}-x_{k}^{2}}{r^{3}}V+\frac{x_{k}}{r}V_{r}\frac{x_{k}}{r}\right)
\right] \\
&  =\nabla\left[  \mu(\rho)\left(  \frac{(N-1)r^{2}}{r^{3}}V+V_{r}\right)
\right] \\
&  =\nabla\left[  \mu(\rho)\left(  \frac{N-1}{r}V+V_{r}\right)  \right] \\
&  =\nabla\mu(\rho)(\frac{N-1}{r}V+V_{r})+\mu(\rho)\nabla(\frac{N-1}{r}%
V+V_{r})\\
&  =\frac{\partial}{\partial x_{i}}\mu_{r}\left(  \rho\right)  \left(
\frac{N-1}{r}V+V_{r}\right)  +\mu(\rho)\frac{\partial}{\partial x_{i}}%
(\frac{N-1}{r}V+V_{r})\\
&  =\frac{\partial r}{\partial x_{i}}\mu_{r}\left(  \rho\right)  \left(
\frac{N-1}{r}V+V_{r}\right)  +\mu(\rho)\left(  -\left(  \frac{N-1}{r^{2}%
}\right)  \frac{\partial r}{\partial x_{i}}V+\frac{N-1}{r}\frac{\partial
r}{\partial x_{i}}V_{r}+\frac{\partial r}{\partial x_{i}}V_{rr}\right) \\
&  =\frac{x_{i}}{r}\mu_{r}\left(  \rho\right)  \left(  \frac{N-1}{r}%
V+V_{r}\right)  +\mu(\rho)\left(  -\left(  \frac{N-1}{r^{2}}\right)
\frac{x_{i}}{r}V+\frac{N-1}{r}\frac{x_{i}}{r}V_{r}+\frac{x_{i}}{r}%
V_{rr}\right) \\
&  =\frac{x_{i}}{r}\left\{  \mu_{r}\left(  \rho\right)  \left(  \frac{N-1}%
{r}V+V_{r}\right)  +\mu(\rho)\left[  -\left(  \frac{N-1}{r^{2}}\right)
V+\frac{N-1}{r}V_{r}+V_{rr}\right]  \right\}
\end{align}
The equations (\ref{pp1}) can be converted to the radial solutions to have%
\begin{align}
&  \rho\left[  u_{t}+\left(  u\cdot\nabla\right)  u\right]  +\nabla
P(\rho)+\rho\nabla\Phi\\
&  =\rho\left(  \left(  u_{i}\right)  _{t}+\overset{n}{\underset{k=1}{%
{\textstyle\sum}
}}u_{k}\frac{\partial u_{i}}{\partial x_{k}}\right)  +\frac{\partial P(\rho
)}{\partial x_{i}}+\rho\frac{\partial\Phi\left(  \rho\right)  }{\partial
x_{i}}\text{, for }i=1,2,...N\\
&  =\rho\left(  \left(  u_{i}\right)  _{t}+u_{i}\frac{\partial u_{i}}{\partial
x_{i}}+\underset{k\neq i}{%
{\textstyle\sum}
}u_{k}\frac{\partial u_{i}}{\partial x_{k}}\right)  +\frac{\partial P(\rho
)}{\partial x_{i}}+\rho\frac{\partial\Phi\left(  \rho\right)  }{\partial
x_{i}}\\
&  =\rho\left[  \frac{\partial}{\partial t}\left(  \frac{x_{i}}{r}V\right)
+\frac{x_{i}}{r}V\frac{\partial}{\partial x_{i}}\left(  \frac{x_{i}}%
{r}\ V\right)  +\underset{k\neq i}{%
{\textstyle\sum}
}\frac{x_{k}}{r}\ V\frac{\partial}{\partial x_{k}}\left(  \frac{x_{i}}%
{r}\ V\right)  \right]  +\left(  \frac{\partial P(\rho)}{\partial r}+\rho
\frac{\partial\Phi\left(  \rho\right)  }{\partial r}\right)  \frac{\partial
r}{\partial x_{i}}\\
&  =\rho\left[  \frac{x_{i}}{r}V_{t}+\frac{x_{i}}{r}V\left(  \frac{r^{2}%
-x_{i}^{2}}{r^{3}}V+\frac{x_{i}}{r}V_{r}\frac{x_{i}}{r}\right)  \right]
+\underset{k\neq i}{%
{\textstyle\sum}
}\frac{x_{k}}{r}V\left(  \frac{-x_{i}x_{k}}{r^{3}}V+\frac{x_{i}x_{k}}{r^{2}%
}V_{r}\right)  +\left[  P_{r}(\rho)+\rho\Phi_{r}\left(  \rho\right)  \right]
\frac{x_{i}}{r}\\
&  =\frac{x_{i}}{r}\left\{  \rho\left[  V_{t}+\left(  \frac{r^{2}-x_{i}^{2}%
}{r^{3}}-\underset{k\neq i}{%
{\textstyle\sum}
}\frac{x_{k}^{2}}{r^{3}}\right)  V^{2}+\left(  \frac{x_{i}^{2}}{r^{3}%
}+\underset{k\neq i}{%
{\textstyle\sum}
}\frac{x_{k}^{2}}{r^{3}}\right)  VV_{r}\right]  +P_{r}(\rho)+\rho\Phi
_{r}\left(  \rho\right)  \right\} \\
&  =\frac{x_{i}}{r}\left[  \rho\left(  V_{t}+VV_{r}\right)  +P_{r}(\rho
)+\rho\Phi_{r}\left(  \rho\right)  \right]  .
\end{align}
Therefore, the Euler-Poisson equations in radial symmetry can be written in
the following form%
\begin{equation}
\left\{
\begin{array}
[c]{c}%
\rho_{t}+V\rho_{r}+\rho V_{r}+\dfrac{N-1}{r}\rho V=0,\\
\rho\left(  V_{t}+VV_{r}\right)  +P_{r}(\rho)+\rho\Phi_{r}\left(  \rho\right)
=\mu_{r}\left(  \rho\right)  \left(  \frac{N-1}{r}V+V_{r}\right)  +\mu
(\rho)\left[  -\left(  \frac{N-1}{r^{2}}\right)  V+\frac{N-1}{r}V_{r}%
+V_{rr}\right]  .
\end{array}
\right.
\end{equation}
Here we under a common assumption for:
\begin{equation}
\mu(\rho):=\kappa\rho^{\theta},
\end{equation}
and $\kappa$ and $\theta\geq0$ are the constants. For the study of this kind
of the above system, the readers may refer \cite{LL}, \cite{MV}, \cite{Ni},
\cite{YZ} and \cite{Y1}. In particular, when $\theta=0$, it returns the
expression for the $V$-dependent only viscosity function:%
\begin{equation}
\left\{
\begin{array}
[c]{c}%
\rho_{t}+V\rho_{r}+\rho V_{r}+\dfrac{N-1}{r}\rho V=0,\\
\rho\left(  V_{t}+VV_{r}\right)  +P_{r}(\rho)+\rho\Phi_{r}\left(  \rho\right)
=\mu(\rho)\left[  V_{rr}+\frac{N-1}{r}V_{r}-\left(  \frac{N-1}{r^{2}}\right)
V\right]  ,
\end{array}
\right.
\end{equation}
which are the usual form Navier-Stokes-Poisson equations. The equations
(\ref{Euler-Poisson})$_{1}$ and (\ref{Euler-Poisson})$_{2}$ $(vis(\rho
,u)\neq0)$ are the compressible Navier-Stokes equations with forcing term. The
equation (\ref{Euler-Poisson})$_{3}$ is the Poisson equation through which the
gravitational potential is determined by the density distribution of the
density itself. Thus, we call the system (\ref{Euler-Poisson}) the
Navier--Stokes-Poisson equations.

Here, if the $vis(\rho,u)=0$, the system is called the Euler-Poisson
equations.\ In this case, the equations can be viewed as a prefect gas model.
For $N=3$, (\ref{Euler-Poisson}) is a classical (nonrelativistic) description
of a galaxy, in astrophysics. See \cite{C}, \cite{M1} for a detail about the system.

$P=P(\rho)$\ is the pressure. The $\gamma$-law can be applied on the pressure
$P(\rho)$, i.e.%
\begin{equation}
{\normalsize P}\left(  \rho\right)  {\normalsize =K\rho}^{\gamma}%
:=\frac{{\normalsize \rho}^{\gamma}}{\gamma}, \label{gamma}%
\end{equation}
which is a commonly the hypothesis. The constant $\gamma=c_{P}/c_{v}\geq1$,
where $c_{P}$, $c_{v}$\ are the specific heats per unit mass under constant
pressure and constant volume respectively, is the ratio of the specific heats,
that is, the adiabatic exponent in (\ref{gamma}). With $K=0$, we call that the
system is pressureless.

For the $3$-dimensional case, we are interested in the hydrostatic equilibrium
specified by $u=0$. According to \cite{C}, the ratio between the core density
$\rho(0)$ and the mean density $\overset{\_}{\rho}$ for $6/5<\gamma<2$\ is
given by%
\begin{equation}
\frac{\overset{\_}{\rho}}{\rho(0)}=\left(  \frac{-3}{z}\dot{y}\left(
z\right)  \right)  _{z=z_{0}}%
\end{equation}
where $y$\ is the solution of the Lane-Emden equation with $n=1/(\gamma-1)$,%
\begin{equation}
\ddot{y}(z)+\dfrac{2}{z}\dot{y}(z)+y(z)^{n}=0,\text{ }y(0)=\alpha>0,\text{
}\dot{y}(0)=0,\text{ }n=\frac{1}{\gamma-1},
\end{equation}
and $z_{0}$\ is the first zero of $y(z_{0})=0$. We can solve the Lane-Emden
equation analytically for%
\begin{equation}
y_{anal}(z)\doteq\left\{
\begin{array}
[c]{ll}%
1-\frac{1}{6}z^{2}, & n=0;\\
\dfrac{\sin z}{z}, & n=1;\\
\dfrac{1}{\sqrt{1+z^{2}/3}}, & n=5,
\end{array}
\right.
\end{equation}
and for the other values, only numerical values can be obtained. It can be
shown that for $n<5$, the radius of polytropic models is finite; for $n\geq5$,
the radius is infinite.\newline For the isothermal case ($\gamma=1$), the
corresponding stationary solution is the Liouville equation,%
\begin{equation}
\ddot{y}(z){\normalsize +}\dfrac{2}{z}\dot{y}(z){\normalsize +\dfrac
{\alpha(3)}{K}e}^{y(z)}=0.
\end{equation}
Gambin \cite{G} and Bezard \cite{B} obtained the existence results about the
explicitly stationary solution $\left(  u=0\right)  $ for $\gamma=6/5$ in
Euler-Poisson equations$:$%
\begin{equation}
\rho=\left(  \frac{3KA^{2}}{2\pi}\right)  ^{5/4}\left(  1+A^{2}r^{2}\right)
^{-5/2}, \label{stationsoluionr=6/5}%
\end{equation}
where $A$ is constant.

In the following, we always seek solutions in radial symmetry. Thus, the
Poisson equation (\ref{Euler-Poisson})$_{3}$ is transformed to%
\begin{equation}
{\normalsize r^{N-1}\Phi}_{rr}\left(  {\normalsize t,x}\right)  +\left(
N-1\right)  r^{N-2}\Phi_{r}{\normalsize =}\alpha\left(  N\right)
{\normalsize \rho r^{N-1},}%
\end{equation}%
\begin{equation}
\Phi_{r}=\frac{\alpha\left(  N\right)  }{r^{N-1}}\int_{0}^{r}\rho
(t,s)s^{N-1}ds.
\end{equation}
In this paper, we concern about blowup solutions for the $N$-dimensional
pressureless Navier-Stokes-Poisson equations with the density-dependent
viscosity. And our aim is to construct a family of such blowup solutions.

Historically in astrophysics, Goldreich and Weber \cite{GW} constructed the
analytical blowup solutions (collapsing) of the $3$-dimensional Euler-Poisson
equations for $\gamma=4/3$ for the non-rotating gas spheres. After that,
Makino \cite{M1} obtained the rigorously mathematical proof of the existence
of such kind of blowup solutions. And in \cite{DXY}, Deng, Xiang and Yang
extended the above blowup solutions in $R^{N}$ ($N\geq4$). Recently, Yuen
obtained the blowup solutions in $R^{2}$ with $\gamma=1$ in \cite{Y1}. The
family of the analytical solutions are rewritten as

For $N\geq3$ and $\gamma=(2N-2)/N$, in \cite{Y}
\begin{equation}
\left\{
\begin{array}
[c]{c}%
\rho(t,r)=\left\{
\begin{array}
[c]{c}%
\dfrac{1}{a(t)^{N}}y(\frac{r}{a(t)})^{N/(N-2)},\text{ for }r<a(t)Z_{\mu};\\
0,\text{ for }a(t)Z_{\mu}\leq r.
\end{array}
\right.  \text{, }V{\normalsize (t,r)=}\dfrac{\dot{a}(t)}{a(t)}%
{\normalsize r,}\\
\ddot{a}(t){\normalsize =}-\dfrac{\lambda}{a(t)^{N-1}},\text{ }%
{\normalsize a(0)=a}_{0}>0{\normalsize ,}\text{ }\dot{a}(0){\normalsize =a}%
_{1},\\
\ddot{y}(z){\normalsize +}\dfrac{N-1}{z}\dot{y}(z){\normalsize +}\dfrac
{\alpha(N)}{(2N-2)K}{\normalsize y(z)}^{N/(N-2)}{\normalsize =\mu,}\text{
}y(0)=\alpha>0,\text{ }\dot{y}(0)=0,
\end{array}
\right.  \label{solution2}%
\end{equation}
where $\mu=[N(N-2)\lambda]/(2N-2)K$ and the finite $Z_{\mu}$ is the first zero
of $y(z)$;

For $N=2$ and $\gamma=1$, in \cite{Y1}%
\begin{equation}
\left\{
\begin{array}
[c]{c}%
\rho(t,r)=\dfrac{1}{a(t)^{2}}e^{y(r/a(t))}\text{, }V{\normalsize (t,r)=}%
\dfrac{\dot{a}(t)}{a(t)}{\normalsize r;}\\
\ddot{a}(t){\normalsize =}-\dfrac{\lambda}{a(t)},\text{ }{\normalsize a(0)=a}%
_{0}>0{\normalsize ,}\text{ }\dot{a}(0){\normalsize =a}_{1};\\
\ddot{y}(z){\normalsize +}\dfrac{1}{z}\dot{y}(z){\normalsize +\dfrac
{\alpha(2)}{K}e}^{y(z)}{\normalsize =\mu,}\text{ }y(0)=\alpha,\text{ }\dot
{y}(0)=0,
\end{array}
\right.  \label{solution 3}%
\end{equation}
where $K>0$, $\mu=2\lambda/K$ with a sufficiently small $\lambda$ and $\alpha$
are constants.\newline For the construction of special analytical solutions to
the Navier-Stokes equations in $R^{N}$ without force term, readers may refer
Yuen's recent results in \cite{Y2}.\newline In this article, the analytical
blowup solutions are constructed in the pressureless Euler-Poisson equations
with density-dependent viscosity in $R^{N}$ in radial symmetry:%
\begin{equation}
\left\{
\begin{array}
[c]{rl}%
\rho_{t}+V\rho_{r}+\rho V_{r}+{\normalsize \dfrac{N-1}{r}\rho V} &
{\normalsize =0,}\\
\rho\left(  V_{t}+VV_{r}\right)  & {\normalsize =-}\dfrac{\alpha(N)\rho
}{r^{N-1}}%
{\displaystyle\int_{0}^{r}}
\rho(t,s)s^{N-1}ds+[\kappa\rho^{\theta}]_{r}\left(  \frac{N-1}{r}%
V+V_{r}\right)  +(\kappa\rho^{\theta})(V_{rr}+\dfrac{N-1}{r}V_{r}+\dfrac
{N-1}{r^{2}}V),
\end{array}
\right.  \label{gamma=1}%
\end{equation}
in the form of the following theorem.

\begin{theorem}
\label{thm:1}For the $N$-dimensional pressureless Navier-Stokes-Poisson
equations in radial symmetry, (\ref{gamma=1}), there exists a family of
solutions:\newline for $N\geq2$ and $N\neq3$, with $\theta=(2N-3)/N$, we
have,
\begin{equation}
\left\{
\begin{array}
[c]{c}%
\rho(t,r)=\left\{
\begin{array}
[c]{c}%
\dfrac{1}{(T-Ct)^{N}}y(\frac{r}{T-Ct})^{\frac{N}{N-3}},\text{ for }\frac
{r}{T-Ct}<Z_{0};\\
0,\text{ for }Z_{0}\leq\frac{r}{T-ct}.
\end{array}
\right.  ,\text{ }V{\normalsize (t,r)=}\dfrac{-C}{T-Ct}{\normalsize r;}\\
\ddot{y}(z){\normalsize +}\dfrac{N-1}{z}\dot{y}(z){\normalsize +\dfrac
{\alpha(N)(N-3)}{(2N-3)CN\kappa}y(z)}^{N/(N-3)}{\normalsize =0,}\text{
}y(0)=\alpha>0,\text{ }\dot{y}(0)=0
\end{array}
\right.  ; \label{eq11123}%
\end{equation}
for $N=3$, with $\theta=1$, we have,%
\begin{equation}
\left\{
\begin{array}
[c]{c}%
\rho(t,r)=\dfrac{1}{(T-Ct)^{3}}e^{y(r/(T-Ct))}\text{, }V{\normalsize (t,r)=}%
\dfrac{-C}{T-Ct}{\normalsize r;}\\
\ddot{y}(z){\normalsize +}\dfrac{2}{z}\dot{y}(z){\normalsize +\dfrac{4\pi
}{CN\kappa}e}^{y(z)}{\normalsize =0,}\text{ }y(0)=\alpha,\text{ }\dot{y}(0)=0,
\end{array}
\right.  \label{eq11123A}%
\end{equation}
where $T>0$, $\kappa>0$, $C\neq0$ and $\alpha$ are constants, and the finite
$Z_{0}$ is the first zero of $y(z)$. In particular, for $C>0$, the solutions
blow up in the finite time $T/C$.
\end{theorem}

\section{Separable Blowup Solutions}

Before presenting the proof of Theorem \ref{thm:1}, we prepare some lemmas.
First, we obtain the solutions for the continuity equation of mass in radial
symmetry (\ref{gamma=1})$_{1}$.

\begin{lemma}
\label{lem:generalsolutionformasseq}For the N-dimensional conservation of mass
in radial symmetry
\begin{equation}
\rho_{t}+V\rho_{r}+\rho V_{r}+\dfrac{N-1}{r}\rho V=0,
\label{massequationspherical}%
\end{equation}
there exist solutions,%
\begin{equation}
\rho(t,r)=\frac{f(r/(T-Ct))}{(T-Ct)^{N}},\text{ }V{\normalsize (t,r)=}%
\frac{-C}{T-Ct}{\normalsize r,} \label{generalsolutionformassequation}%
\end{equation}
where $f\geq0\in C^{1}$, $T$ and $C$ are positive constants.
\end{lemma}

\begin{proof}
We just plug (\ref{generalsolutionformassequation}) into
(\ref{massequationspherical}). Then
\begin{align}
&  \rho_{t}+V\rho_{r}+\rho V_{r}+\frac{N-1}{r}\rho V\\
&  =\frac{(-N)(-C)f(r/(T-Ct))}{(T-Ct)^{N+1}}-\frac{(-C)r\overset{\cdot}%
{f}(r/(T-Ct))}{(T-Ct)^{N+2}}\\
&  +\frac{(-C)r}{T-Ct}\frac{\overset{\cdot}{f}(r/(T-Ct))}{(T-Ct)^{N+1}}%
+\frac{f(r/(T-Ct))}{(T-Ct)^{N}}\frac{(-C)}{T-Ct}+\frac{N-1}{r}\frac
{f(r/(T-Ct))}{(T-Ct)^{N}}\frac{(-C)}{T-Ct}r\\
&  =0.
\end{align}
The proof is completed.
\end{proof}

Besides, we need the lemma for stating the property of the function $y(z)$ of
the analytical solutions (\ref{eq11123}) and (\ref{eq11123A}). We need the
lemma for stating the property of the function $y(z)$. In particular, the
solutions (\ref{eq11123}) in 2-dimensional case involve the following lemma.
The similar lemma was already given in Lemmas 9 and 10, in \cite{Y1}, by the
fixed point theorem. For the completeness of understanding the whole article,
the proof is also presented here.

\begin{lemma}
\label{lemma2AA}For the ordinary differential equation,%
\begin{equation}
\left\{
\begin{array}
[c]{c}%
\ddot{y}(z){\normalsize +}\dfrac{1}{z}\dot{y}(z)-{\normalsize \frac{\sigma
}{y(z)^{2}}=0}\\
y(0)=\alpha>0,\text{ }\dot{y}(0)=0,
\end{array}
\right.  \label{SecondorderElliptic1}%
\end{equation}
where $\sigma$ is a positive constant,\newline has a solution $y(z)\in C^{2}$
and $\underset{z\rightarrow+\infty}{\lim}y(z)=\infty$.
\end{lemma}

\begin{proof}
By integrating (\ref{SecondorderElliptic1}), we have,%
\begin{equation}
\overset{\cdot}{y}(z)=\frac{\sigma}{z}\int_{0}^{z}\frac{1}{y(s)^{2}}sds\geq0.
\label{lemma3eq1}%
\end{equation}
Thus, for $0<z<z_{0}$, $y(z)$ has a uniform lower upper bound
\begin{equation}
y(z)\geq y(0)=\alpha>0.
\end{equation}
As we obtained the local existence in Lemma \ref{lemma2AA}, there are two
possibilities:\newline(1)$y(z)$ only exists in some finite interval
$[0,z_{0}]$: (1a)$\underset{z\rightarrow z_{0-}}{\lim}y(z)=\infty$; (1b)$y(z)$
has an uniformly upper bound, i.e. $y(z)\leq\alpha_{0}$ for some constant
$\alpha_{0}.$\newline(2)$y(z)$ exists in $[0,$ $+\infty)$: (2a)$\underset
{z\rightarrow+\infty}{\lim}y(z)=\infty$; (2b)$y(z)$ has an uniformly upper
bound, i.e. $y(z)\leq\beta$ for some positive constant $\beta$.\newline We
claim that possibility (1) does not exist. We need to reject (1b) first: If
the statement (1b) is true, (\ref{lemma3eq1}) becomes%
\begin{equation}
\frac{\sigma z}{2\alpha^{2}}=\frac{\sigma}{z}\int_{0}^{z}\frac{s}{\alpha^{2}%
}ds\geq\overset{\cdot}{y}(z). \label{possible1}%
\end{equation}
Thus, $\overset{\cdot}{y}(z)$ is bounded in $[0,z_{0}]$. Therefore, we can use
the fixed point theorem again to obtain a large domain of existence, such that
$[0,z_{0}+\delta]$ for some positive number $\delta$. There is a
contradiction. Therefore, (1b) is rejected.\newline Next, we do not accept
(1a) because of the following reason: It is impossible that $\underset
{z\rightarrow z_{0-}}{\lim}y(z)=\infty$, as from (\ref{possible1}),
$\overset{\cdot}{y}(z)$ has an upper bound in $[0,$ $z_{0}]$:%
\begin{equation}
\frac{\sigma z_{0}}{2\alpha^{2}}\geq\overset{\cdot}{y}(z). \label{lemma3eq2}%
\end{equation}
Thus, (\ref{lemma3eq2}) becomes,
\begin{equation}
y(z_{0})=y(0)+\int_{0}^{z_{0}}\overset{\cdot}{y}(s)ds\leq\alpha+\int
_{0}^{z_{0}}\frac{\sigma z_{0}}{2\alpha^{2}}ds=\alpha+\frac{\sigma(z_{0})^{2}%
}{2\alpha^{2}}%
\end{equation}
Since $y(z)$ is bounded above in $[0,$ $z_{0}]$, it contracts the statement
(1a), such that $\underset{z\rightarrow z_{0-}}{\lim}y(z)=\infty$. So, we can
exclude the possibility (1).\newline We claim that the possibility (2b)
doesn't exist. It is because
\begin{equation}
\overset{\cdot}{y}(z)=\frac{\sigma}{z}\int_{0}^{z}\frac{s}{y(s)^{2}}%
ds\geq\frac{\sigma}{z}\int_{0}^{z}\frac{s}{\beta^{2}}ds=\frac{\sigma z}%
{2\beta^{2}}.
\end{equation}
Then, we have,%
\begin{equation}
y(z)\geq\alpha+\frac{\sigma}{4\beta^{2}}z^{2}. \label{lemma3eq3}%
\end{equation}
By letting $z\rightarrow\infty$, (\ref{lemma3eq3}) turns out to be,
\begin{equation}
y(z)=\infty.
\end{equation}
Since a contradiction is established, we exclude the possibility (2b). Thus,
the equation (\ref{SecondorderElliptic1}) exists in $[0,$ $+\infty)$ and
$\underset{z\rightarrow+\infty}{\lim}y(z)=\infty$. This completes the proof.
\end{proof}

For $N\geq4$, our blowup solutions depend on the Lane-Emden equation,
\begin{equation}
\left\{
\begin{array}
[c]{c}%
\ddot{y}(z)+\dfrac{N-1}{z}\overset{\cdot}{y}(z)+\sigma y(z)^{s}%
{\normalsize =0,}\\
y(0)=\alpha>0,\text{ }\dot{y}(0)=0,
\end{array}
\right.  \label{eee}%
\end{equation}
where $s>1$ and $\sigma$ are positive constants, which is a particular case of
the Emden-Fowler equation%
\begin{equation}
\ddot{u}+z^{1-n}u^{n}=0,
\end{equation}
with $n>1$, after the transformation
\begin{equation}
u=\frac{y(z)}{z}.
\end{equation}
For the existence and uniqueness of the equation (\ref{eee}), the interested
readers may refer the survey paper \cite{Wong}.

The proof of Theorem \ref{thm:1} is similar to the ones in \cite{DXY},
\cite{Y} and \cite{Y1}. The main idea is to put the analytical solutions to
check that if they satisfy the system (\ref{gamma=1}) only.

\begin{proof}
[Proof of Theorem \ref{thm:1}]From Lemma \ref{lem:generalsolutionformasseq},
it is clear for that (\ref{eq11123}) satisfy (\ref{gamma=1})$_{1}$.\newline
For the case of $N\geq2$ and $N\neq3$ with $\theta=(2N-3)/N$, we plug the
solutions (\ref{eq11123}) into the momentum equation (\ref{gamma=1})$_{2}$,%
\begin{align}
&  \rho(V_{t}+VV_{r})+\frac{\alpha(N)\rho}{r^{N-1}}%
{\displaystyle\int\limits_{0}^{r}}
\rho(t,s)s^{N-1}ds-[\kappa\rho^{(2N-3)/N}]_{r}\left(  \frac{N-1}{r}%
V+V_{r}\right)  -\kappa\rho^{(2N-3)/N}(V_{rr}+\dfrac{N-1}{r}V_{r}-\dfrac
{N-1}{r^{2}}V)\\
&  =\rho\left[  \frac{(-C)(-1)(-C)}{(T-Ct)^{2}}r+\frac{(-C)}{T-Ct}r\cdot
\frac{(-C)}{T-Ct}\right]  +\frac{\alpha(N)\rho}{r^{N-1}}%
{\displaystyle\int\limits_{0}^{r}}
\frac{y(\frac{s}{T-Ct})^{N/(N-3)}}{(T-Ct)^{N}}s^{N-1}ds\\
&  -\frac{\partial}{\partial r}N\kappa\left[  \frac{1}{(T-Ct)^{N}}y\left(
\frac{r}{T-Ct}\right)  ^{N/(N-3)}\right]  _{{}}^{(2N-3)/N}\frac{(-C)}%
{T-Ct}-0\\
&  =\frac{\alpha(N)\rho}{r^{N-1}}%
{\displaystyle\int\limits_{0}^{r}}
\frac{y(\frac{s}{T-Ct})^{N/(N-3)}}{(T-Ct)^{N}}s^{N-1}ds+\frac{(2N-3)CN\kappa
}{N-3}\frac{y\left(  \frac{r}{T-Ct}\right)  }{(T-Ct)^{N-3}}\frac{y\left(
\frac{r}{T-Ct}\right)  ^{[N/(N-3)]-1}\overset{\cdot}{y}\left(  \frac{r}%
{T-Ct}\right)  }{(T-Ct)^{N+2}}\\
&  =\frac{\alpha(N)\rho}{r^{N-1}}%
{\displaystyle\int\limits_{0}^{r}}
\frac{y(\frac{s}{T-Ct})^{N/(N-3)}}{(T-Ct)^{N}}s^{N-1}ds+\frac{(2N-3)CN\kappa
}{N-3}\frac{\rho}{(T-Ct)^{N-1}}\overset{\cdot}{y}\left(  \frac{r}{T-Ct}\right)
\\
&  =\frac{\rho}{(T-Ct)^{N-1}}\left[  \frac{(2N-3)CN\kappa}{N-3}\overset{\cdot
}{y}(\frac{r}{T-Ct})+\frac{\alpha(N)}{r^{N-1}(T-Ct)}%
{\displaystyle\int\limits_{0}^{r}}
y(\frac{s}{T-Ct})^{N/(N-3)}s^{N-1}ds\right] \\
&  =\frac{\rho}{(T-Ct)^{N-1}}\left[  \frac{(2N-3)CN\kappa}{N-3}\overset{\cdot
}{y}(\frac{r}{T-Ct})+\frac{\alpha(N)}{(\frac{r}{T-Ct})^{N-1}}%
{\displaystyle\int\limits_{0}^{r/(T-Ct)}}
y(s)^{N/(N-3)}s^{N-1}ds\right] \\
&  =\frac{\rho}{(T-Ct)^{N-1}}Q\left(  \frac{r}{T-Ct}\right)  .
\end{align}
And denote%
\begin{equation}
Q(\frac{r}{T-Ct})\doteq{\normalsize Q(z)=}\frac{(2N-3)CN\kappa}{N-3}%
\overset{\cdot}{y}(z){\normalsize +}\frac{\alpha(N)}{z^{N-1}}%
{\displaystyle\int\limits_{0}^{z}}
y{\normalsize (s)}^{N/(N-3)}{\normalsize s}^{N-1}{\normalsize ds.}%
\end{equation}
Differentiate $Q(z)$\ with respect to $z$,%
\begin{align}
\overset{\cdot}{Q}(z)  &  =\frac{(2N-3)CN\kappa}{N-3}\overset{\cdot\cdot}%
{y}(z){\normalsize +}\alpha(N)y{\normalsize (z)}^{N/(N-3)}-\frac
{(N-1)\alpha(N)}{z^{N-1}}%
{\displaystyle\int\limits_{0}^{z}}
y{\normalsize (s)}^{N/(N-3)}{\normalsize s}^{N-1}{\normalsize ds}\\
&  =-\left(  \frac{N-1}{z}\right)  \frac{(2N-3)CN\kappa}{N-3}%
{\normalsize \overset{\cdot}{y}(z)-}\frac{(N-1)}{z}\frac{\alpha(N)}{z^{N-1}}%
{\displaystyle\int\limits_{0}^{z}}
y{\normalsize (s)}^{N/(N-3)}{\normalsize s}^{N-1}{\normalsize ds}\\
&  =\frac{-(N-1)}{z}Q(z),
\end{align}
where the above result is due to the fact that we choose the following
ordinary differential equation,%
\begin{equation}
\left\{
\begin{array}
[c]{c}%
\ddot{y}(z)+\dfrac{N-1}{z}\overset{\cdot}{y}(z)+\dfrac{\alpha(N)(N-3)}%
{(2N-3)CN\kappa}y(z)^{N/(N-3)}=0.\\
{\normalsize y(0)=\alpha>0,}\text{ }\dot{y}(0){\normalsize =0,}%
\end{array}
\right.
\end{equation}
in Lemma \ref{lemma2AA} and the well-known results in the Emden-Folwer
equation. With $Q(0)=0$, this implies that $Q(z)=0$. Thus, the momentum
equation (\ref{gamma=1})$_{2}$ is satisfied.\newline For the case of $N=3$
with $\theta=1$, we plug the solutions (\ref{eq11123A}) into the momentum
equation (\ref{gamma=1})$_{2}$,%
\begin{align}
&  \rho(V_{t}+VV_{r})+\frac{4\pi\rho}{r^{2}}%
{\displaystyle\int\limits_{0}^{r}}
\rho(t,s)s^{2}ds-[\mu(\rho)]_{r}\left(  \frac{N-1}{r}V+V_{r}\right)  -\mu
(\rho)(V_{rr}+\dfrac{2}{r}V_{r}-\dfrac{2}{r^{2}}V)\\
&  =\rho\left[  \frac{(-C)(-1)(-C)}{(T-Ct)^{2}}r+\frac{(-C)}{T-Ct}r\cdot
\frac{(-C)}{T-Ct}\right]  +\frac{4\pi\rho}{r^{2}}%
{\displaystyle\int\limits_{0}^{r}}
\frac{e^{y(s/(T-Ct))}}{(T-Ct)^{3}}s^{2}ds\\
&  -(\kappa\rho)_{r}\frac{N(-C)}{T-Ct}-\mu(\rho)\left(  0+\dfrac{2}{r}%
\frac{(-1)}{T-Ct}-\frac{2}{r^{2}}\frac{(-1)}{T-Ct}r\right) \\
&  =\frac{\kappa e^{y(r/(T-Ct))}}{(T-Ct)^{3}}\cdot\dot{y}(\frac{r}{T-Ct}%
)\cdot\frac{1}{T-Ct}\frac{NC}{T-Ct}+\frac{4\pi\rho}{r^{2}}%
{\displaystyle\int\limits_{0}^{r}}
\frac{e^{y(s/(T-Ct))}}{(T-Ct)^{3}}s^{2}ds\\
&  =\frac{CN\kappa\rho}{(T-Ct)^{2}}\dot{y}(\frac{r}{T-Ct})+\frac{4\pi\rho
}{r^{2}}%
{\displaystyle\int\limits_{0}^{r}}
\frac{e^{y(s/(T-Ct))}}{(T-Ct)^{3}}s^{2}ds\\
&  =\frac{\rho}{(T-Ct)^{2}}\left[  CN\kappa\dot{y}(\frac{r}{T-Ct})+\frac{4\pi
}{r^{2}(T-Ct)}%
{\displaystyle\int\limits_{0}^{r}}
e^{y\left(  s/(T-Ct)\right)  }s^{2}ds\right]  .
\end{align}
By letting $\omega=s/(T-Ct)$, it follows:%
\begin{align}
&  =\frac{\rho}{(T-Ct)^{2}}\left[  CN\kappa\dot{y}(\frac{r}{T-Ct})+\frac{4\pi
}{(\frac{r}{T-Ct})^{2}}%
{\displaystyle\int\limits_{0}^{r/(T-Ct)}}
e^{y\left(  s\right)  }s^{2}ds\right] \\
&  =\frac{\rho}{(T-Ct)^{2}}Q\left(  \frac{r}{T-Ct}\right)  .
\end{align}
And denote $z=r/(T-Ct)$,%
\begin{equation}
Q(\frac{r}{T-Ct})={\normalsize Q(z)=}CN\kappa\dot{y}(z){\normalsize +}%
\dfrac{4\pi}{z^{2}}%
{\displaystyle\int\limits_{0}^{z}}
e^{y\left(  \omega\right)  }\omega^{2}d\omega{\normalsize .}%
\end{equation}
Differentiate $Q(z)$\ with respect to $z$,%
\begin{align}
\dot{Q}(z)  &  =CN\kappa\ddot{y}(z){\normalsize +4\pi e}^{y(z)}+\frac
{(-2)\cdot4\pi}{z^{3}}%
{\displaystyle\int\limits_{0}^{z}}
e^{y(\omega)}{\normalsize \omega}^{2}{\normalsize d}\omega\\
&  =-\frac{2CN\kappa}{z}\dot{y}(z)+\frac{(-2)4\pi}{z^{3}}%
{\displaystyle\int\limits_{0}^{z}}
e^{y(\omega)}{\normalsize \omega}^{2}{\normalsize d\omega}\\
&  =-\frac{2}{z}Q(z),
\end{align}
where the above result is due to the fact that we choose the Liouville
equation:%
\begin{equation}
\left\{
\begin{array}
[c]{c}%
\ddot{y}(z){\normalsize +}\dfrac{2}{z}\dot{y}(z){\normalsize +\dfrac{4\pi
}{CN\kappa}e}^{y(z)}{\normalsize =0}\\
{\normalsize y(0)=\alpha,}\text{ }\dot{y}(0){\normalsize =0.}%
\end{array}
\right.
\end{equation}
With $Q(0)=0$, this implies that $Q(z)=0$. Now we are able to show that the
family of the solutions blows up, in the finite time $T/C$. This completes the proof.
\end{proof}

The statement about the blowup rate will be immediately followed:

\begin{corollary}
The blowup rate of the solution (\ref{eq11123}) and (\ref{eq11123A}) is
\begin{equation}
\underset{t\rightarrow\frac{T}{C}^{-}}{\lim}\rho(t,0)(T-Ct)^{N}\geq O(1).
\end{equation}

\end{corollary}

\begin{remark}
Besides, if we consider the Navier-Stokes equations with the repulsive force
in radial symmetry in $R^{N}$ $(N\neq1)$,%
\begin{equation}
\left\{
\begin{array}
[c]{rl}%
\rho_{t}+V\rho_{r}+\rho V_{r}+{\normalsize \dfrac{N-1}{r}\rho V} &
{\normalsize =0,}\\
\rho\left(  V_{t}+VV_{r}\right)  & {\normalsize =+}\dfrac{\alpha(N)\rho
}{r^{N-1}}%
{\displaystyle\int_{0}^{r}}
\rho(t,s)s^{N-1}ds+[\kappa\rho^{\theta}]_{r}\left(  \frac{N-1}{r}%
V+V_{r}\right)  +(\kappa\rho^{\theta})(V_{rr}+\dfrac{N-1}{r}V_{r}-\dfrac
{N-1}{r^{2}}V)
\end{array}
\right.
\end{equation}
the special solutions are:\newline for $N\geq2$ and $N\neq3$, with
$\theta=(2N-3)/N$,
\begin{equation}
\left\{
\begin{array}
[c]{c}%
\rho(t,r)=\left\{
\begin{array}
[c]{c}%
\dfrac{1}{(T-Ct)^{N}}y(\frac{r}{T-Ct})^{\frac{N}{N-3}},\text{ for }\frac
{r}{T-Ct}<Z_{0};\\
0,\text{ for }Z_{0}\leq\frac{r}{T-Ct}%
\end{array}
\right.  ,\text{ }V{\normalsize (t,r)=}\dfrac{-C}{T-Ct}{\normalsize r;}\\
\ddot{y}(z){\normalsize +}\dfrac{N-1}{z}\dot{y}(z)-{\normalsize \dfrac
{\alpha(N)(N-3)}{(2N-3)CN\kappa}y(z)}^{N/(N-3)}{\normalsize =0,}\text{
}y(0)=\alpha>0,\text{ }\dot{y}(0)=0,
\end{array}
\right.
\end{equation}
the finite $Z_{0}$ is the first zero of $y(z)$;\newline for $N=3$, with
$\theta=1$,%
\begin{equation}
\left\{
\begin{array}
[c]{c}%
\rho(t,r)=\dfrac{1}{(T-Ct)^{3}}e^{y(r/(T-Ct))}\text{, }V{\normalsize (t,r)=}%
\dfrac{-C}{T-Ct}{\normalsize r;}\\
\ddot{y}(z){\normalsize +}\dfrac{2}{z}\dot{y}(z)-{\normalsize \dfrac{4\pi
}{CN\kappa}e}^{y(z)}{\normalsize =0,}\text{ }y(0)=\alpha>0,\text{ }\dot
{y}(0)=0.
\end{array}
\right.
\end{equation}

\end{remark}

\end{document}